\documentclass[a4paper,cleveref, autoref, thm-restate, thm-restate, UKenglish]{lipics-v2021}

\usepackage{hyperref}
\usepackage[ruled,vlined, linesnumbered]{algorithm2e}
\usepackage{mathtools}
\usepackage{tikz}
\usetikzlibrary{shapes.geometric}
\usetikzlibrary {decorations.markings} 

\newcommand{\argmax}{\text{argmax}}

\newcommand{\dist}{d}
\newcommand{\alg}{\textsc{Alg}}
\newcommand{\opt}{\textsc{Opt}}
\newcommand{\mst}{\mathrm{MST}}

\newcommand{\blocking}{\textsc{Blocking}}
\newcommand{\optspan}{\textsc{OptSpan}_\delta}
\newcommand{\cost}{W_{\blocking}}

\renewcommand{\subset}{\subseteq}
\renewcommand{\epsilon}{\varepsilon}

\crefformat{lemma}{Lemma~#2#1#3}
\crefformat{theorem}{Theorem~#2#1#3}
\crefformat{proposition}{\mbox{Proposition~#2#1#3}}
\crefformat{remark}{Remark~#2#1#3}
\crefformat{equation}{(#2#1#3)}
\crefformat{corollary}{Corollary~#2#1#3}
\crefformat{example}{Example~#2#1#3}
\crefformat{definition}{Definition~#2#1#3}
\crefformat{observation}{Observation~#2#1#3}
\crefformat{algorithm}{Algorithm~#2#1#3}

\pdfoutput=1 
\hideLIPIcs 
\nolinenumbers

\title{Exploration of graphs with excluded minors}

\author{Júlia Baligács}{Technische Universität Darmstadt, Germany}{baligacs@mathematik.tu-darmstadt.de}{https://orcid.org/0000-0003-2654-149X}{}

\author{Yann Disser}{Technische Universität Darmstadt, Germany  }{disser@mathematik.tu-darmstadt.de}{https://orcid.org/0000-0002-2085-0454}{}

\author{Irene Heinrich}{Technische Universität Darmstadt, Germany}{heinrich@mathematik.tu-darmstadt.de}{https://orcid.org/0000-0001-9191-1712}{}

\author{Pascal Schweitzer}{Technische Universität Darmstadt, Germany}{schweitzer@mathematik.tu-darmstadt.de}{}{}

\authorrunning{J. Baligács et al.}

\Copyright{Júlia Baligács, Yann Disser, Irene Heinrichs, and Pascal Schweitzer} 



\begin{document}

\begin{titlepage}

\maketitle

\begin{abstract}

We study the online graph exploration problem proposed by Kalyana\-sun\-daram and Pruhs~(1994) and prove a constant competitive ratio on minor-free graphs.
This result encompasses and significantly extends the graph classes that were previously known to admit a constant competitive ratio.
The main ingredient of our proof is that we find a connection between the performance of the particular exploration algorithm $\blocking$ and the existence of light spanners. 
Conversely, we exploit this connection to construct light spanners of bounded genus graphs.
In particular, we achieve a lightness that improves on the best known upper bound for genus~$g\geq 1$ and recovers the known tight bound for the planar case ($g=0$).

\end{abstract}

\end{titlepage}
\setcounter{page}{1}
\section{Introduction}
\label{sec:typesetting-summary}

We study a classic online graph exploration problem that was first proposed by Kalyana\-sun\-daram and Pruhs in 1994~\cite{pruhs94}.
In this setting, a single agent needs to systematically traverse an initially unknown, undirected, connected graph with non-negative edge weights.
Upon visiting a new vertex, the agent learns the unique identifiers of all adjacent vertices and the weights of the corresponding edges. The cost incurred when traversing an edge is simply its weight.
The objective in online graph exploration is to visit all vertices of the graph and return to the starting vertex while minimizing the total cost.

The performance of a (deterministic) online algorithm~$\alg$ is measured in terms of competitive analysis.
That is, given a graph $G$ and starting vertex $v$ of $G$, we compare the cost $\alg(G,v)$ of the traversal it produces to the cost of an offline optimum traversal~$\opt(G)$. 
Note that the optimum cost corresponds to the length of a shortest TSP tour of~$G$ and does not depend on $v$. 
We say that~$\alg$ is \emph{(strictly) $\rho$-competitive} for a class of graphs if $\alg(G,v) \leq \rho \cdot \opt (G)$ for every graph $G$ in the class and every vertex $v$ of $G$. 
The \emph{(strict) competitive ratio} of an algorithm $\alg$ is given by $\inf \left\{ \rho : \alg \text{ is $\rho$-competitive} \right\}$. 

Kalyanasundaram and Pruhs~\cite{pruhs94} posed the following question: 
\emph{Is there a deterministic algorithm for online graph exploration with a constant competitive ratio?}
Several algorithms were proposed with a competitive ratio of $\mathcal{O} ( \log (n))$ \cite{megow, NN}, where $n$ is the number of vertices, but better competitive ratios are only known for restricted classes of graphs~\cite{pruhs94, megow, miyazaki}.
The best known lower bound on the competitive ratio is 10/3 \cite{birx}.
In particular, the original question of Kalyanasundaram and Pruhs remains open.

We formalize a connection between the performance of the particular exploration algorithm $\blocking$ and the  existence of light spanners.
Spanners were introduced in 1989 by Peleg and Schäffer~\cite{peleg} and have been instrumental in the development of approximation algorithms, particularly for TSP~\cite{arora, borradaile, minorspanner}.
Here, a subgraph~$H=(V,E_H)$ of a connected, undirected graph~$G=(V,E)$ with edge weights~$w\colon E \to \mathbb{R}_{\geq 0}$ is called a \emph{$(1+\varepsilon)$-spanner} of~$G$ if $d_H(u,v) \leq (1+\varepsilon)\,d_G(u,v)$ for all $u,v \in V$, where $d_H$ and~$d_G$ denote the shortest-path distance in~$H$ and~$G$, respectively.
Then, $H$ has \emph{stretch} at most $(1+\varepsilon)$ and its \emph{lightness} is $w(H)/w(\mst)$, where $w(H):=\sum_{e \in E_H} w(e)$ and $\mst$ denotes a minimum spanning tree of~$G$.

We show that the online graph exploration algorithm $ \blocking$ has a constant competitive ratio on every class of graphs that admits spanners of constant lightness for a fixed stretch.
Prominent graph classes with this property are the classes with a forbidden minor~\cite{minorspanner}.  
We thus, in particular, obtain a constant competitive ratio for online graph exploration on all graph classes excluding a minor.
They encompass many other important classes, such as graphs of bounded genus or bounded treewidth.
Overall, this result subsumes and significantly extends all previously known graph classes for which a competitive ratio of~$o(\log (n))$ was known.

Regarding research for graph spanners, results typically revolve around the existence of good, in particular light, spanners.
For example, the Erd\H{o}s girth conjecture~\cite{erdos} is equivalent to a lower bound of~$\Omega(n^{1/k})$ on the best lightness of a $(2k-1)$-spanner in unweighted graphs.
While this conjecture remains unresolved, a nearly matching upper bound was proven by Chechik and Wulff-Nilsen~\cite{chechik}.
Various constant upper bounds on the lightness are known for restricted classes of graphs~\cite{althoefer,  minorspanner, demainetreewidth, grigni}.
Our newly discovered connection to graph exploration also allows us to contribute an improved upper bound for graphs of bounded genus using the ideas given in \cite{megow}.

\subparagraph*{Our results} 

We significantly expand the class of graphs on which the exploration problem admits a constant-competitive algorithm.

\begin{theorem}\label{thm:explorationminor}
For every graph $H$ and constant $\delta>0$, there is a constant $c$ (depending on~$H$ and $\delta$) such that $\blocking_\delta$ is $c$-competitive on $H$-minor-free graphs.
\end{theorem}

The technical contribution of this work is a new-found connection between graph spanners and the performance of the exploration algorithm $\blocking_\delta$ (see Section \ref{sec:blocking}) introduced by Megow et al.~\cite{megow} based on an algorithm of Kalyanasundaram and Pruhs~\cite{pruhs94}. This connection will allow us to prove Theorem~\ref{thm:explorationminor}.

Prior to our work, the largest class of graphs which was known to admit a constant-competitive algorithm was the class of bounded genus graphs~\cite{megow}. 
As an aside, we obtain a slightly stronger bound also for bounded genus graphs (cf.~Corollary~\ref{cor:blockinggenus}).

So far, $\blocking_\delta$ was only studied for constant choices of the parameter $\delta$, i.e., independent of the number of vertices $n$. It is known that its competitive ratio is at least $\Omega (n^{1/(4+\delta )})$ if $\delta$ is a constant \cite{megow}. This naturally raises the question of whether improvement is possible if $\delta$ may depend on $n$. We obtain the following results. 

\begin{theorem}\label{thm:explorationlogn}
$\blocking_{\log (n)}$ is $O(\log(n))$-competitive.
\end{theorem}

This shows that $\blocking_{\log (n)}$ achieves the best previously known competitiveness. We complement this with the following lower bounds.

\begin{theorem}\label{thm:blockinglower}
The competitive ratio of $\blocking_\delta$, where $\delta=\delta(n)>0$, is at least
\begin{enumerate}[a)]
\item $\Omega(\log (n)/ \log (\log (n)))$,\label{thmpart:general:lowerbound}
\item $\Omega(\log (n))$ for $\delta \in o(\log (n)/ \log \log (n))$ as well as for~$\delta \in \Omega(\log (n))$.\label{thmpart:specific:delta:lowerbound}
\end{enumerate}
\end{theorem}

In particular, this shows that there is no $\delta$ such that $ \blocking_\delta$ is constant-competitive, but it remains open, whether there is a choice of $\delta$ for which the algorithm is $o(\log (n))$-competitive.

Next, we exploit the connection between spanners and exploration in reverse to derive the existence of good spanners in bounded genus graphs.

\newcounter{restatecounter}
\setcounter{restatecounter}{\value{definition}}

\begin{restatable}{theorem}{thmgenusspanner}\label{thm:genusspanner}
For every $\epsilon>0$, the greedy $(1+\epsilon)$-spanner of a graph of genus $g$ has lightness at most $\bigl( 1+\frac{2}{\epsilon} \bigl) \bigl(1+\frac{2g}{1+\epsilon} \bigl)$.
\end{restatable}

Prior to our work, the best known bound was due to Grigni \cite{grigni} who showed that every graph of genus $g\geq 1$ contains a $(1+\epsilon)$-spanner of lightness $1+ \frac{12g-4}{\epsilon}$. Moreover,
it is already known that planar graphs, i.e., graphs of genus 0, contain $(1+\epsilon)$-spanners of lightness $1+\frac{2}{\epsilon}$ and that this is best possible \cite{althoefer}. This means that \cref{thm:genusspanner} gives a tight bound in the case $g=0$ and extrapolates this bound to graphs of larger genus.

\subparagraph*{Related work}
Kalyanasundaram and Pruhs~\cite{pruhs94} introduced the online graph exploration problem and gave a constant-competitive algorithm for planar graphs. 
Megow, Mehlhorn and Schweitzer~\cite{megow} revisited the algorithm, addressed some technical intricacies, and proposed their reinterpretation $\blocking_\delta$, which we also consider in this paper. 
They expanded the result by  Kalyanasundaram and Pruhs and showed that the algorithm is constant-competitive on bounded genus graphs. Moreover, they suggested a new algorithm hDFS and showed that it is constant-competitive on graphs with a bounded number of different weights and $O(\log (n))$-competitive on general graphs.

Another very natural approach for exploration is the \emph{Nearest Neighbor} algorithm, which, in each step, explores the unvisited vertex nearest to the current location. This algorithm has been studied extensively as a TSP heuristic. Rosenkrantz, Stearns and Lewis were able to show that its competitive ratio is $\Theta(\log (n))$ \cite{NN}. It turned out that the lower bound of $\Omega(\log (n))$ is already achieved on unweighted planar graphs \cite{hurkens} and on trees \cite{fritsch}.
Eberle et al.~\cite{eberle} revisited the algorithm with learning augmentation.

In addition to planar and bounded genus graphs, the exploration problem has been studied on many more graph classes. For example, Miyazaki, Morimoto and Okabe were able to show that the competitive ratio of the exploration problem is  $(1+\sqrt{3})/2$ on cycles and 2 on unweighted graphs. Other examples of such graph classes are tadpole graphs \cite{tadpole}, unicyclic graphs \cite{fritsch}, and cactus graphs \cite{fritsch}.

Currently, the best known lower bound for the graph exploration problem is 10/3 which was shown by Birx, Disser, Hopp, and Karousatou~\cite{birx}. Their construction builds on a previously known lower bound of 2.5 shown by Dobrev, Královič, and Markou \cite{dobrev}. Since the construction by Birx et al.~is planar, the lower bound of 10/3 even holds when the problem is restricted to planar graphs.

Several other settings of the exploration problem have been studied, such as exploration on directed graphs \cite{albers,deng, foerster, fleischer} or exploration with a team of agents \cite{dereniowski,DisserHackfeldKlimm/19, disserlower}. Another problem which is closely related to graph exploration is online TSP, where a single agent has to serve requests appearing over time in a known graph \cite{bjelde, bonifaci}.

Through the connection with spanners, we are concerned with the existence of light spanners for a given stretch.
Examples of graph classes where the worst-case lightness does not depend on the number of vertices include planar graphs \cite{althoefer}, bounded genus graphs~\cite{grigni}, apex graphs \cite{grigniapex}, bounded pathwidth graphs \cite{grignipathwidth}, bounded treewidth graphs \cite{demainetreewidth}, and minor-free graphs~\cite{minorspanner}.
Our results rely on the existence of light spanners for minor-free graphs~\cite{minorspanner} and improve on the lightness for bounded genus graphs.
In particular, we study the lightness of the so-called greedy spanner~\cite{althoefer} for graphs of bounded genus. 
It was shown by Filtser and Solomon~\cite{filtser} that this spanner construction is existentially optimal for every class of graphs closed under taking subgraphs, which means that the optimal lightness guarantee on any such class is achieved by the greedy spanner.

Light and sparse spanners have applications in various fields. Most importantly, spanners were used to give polynomial-time approximation schemes (PTAS) for the travelling salesperson problem for various graph classes \cite{arora, borradaile, minorspanner}. 
Note that the difference between approximations for TSP and online exploration is that, in our setting, the tour is computed on-the-fly. 
Indeed, in comparison to our online setting, we desire a constant approximation for an arbitrary constant, which in the TSP setting is easily obtained by traversing a minimum spanning tree twice. On the other hand, in the online setting, we are not concerned with efficiency of the algorithms which is crucial in the TSP setting.

Other fields of application of spanners include distributed systems \cite{distcomp}, routing \cite{routing}, or computational biology \cite{biology}.

\section{The online graph exploration problem on minor-free graphs}

In this section, we prove new upper bounds for $\blocking_\delta$ on  $H$-minor-free graphs (Theorem~\ref{thm:explorationminor}) and for general graphs (Theorem \ref{thm:explorationlogn}). To this end, we begin by introducing the algorithm $\blocking_\delta$ proposed by Megow et~al.~\cite{megow} based on the work of Kalyanasundaram and Pruhs \cite{pruhs94}.

\subsection{The algorithm $\blocking$}\label{sec:blocking}

During the execution of an online graph exploration algorithm, a vertex is \emph{explored} if it has been visited by the agent. A neighbor of an explored vertex is a \emph{learned} vertex.
An edge is a \emph{boundary edge} if exactly one of its endpoints is explored. 
By convention, we denote boundary edges $e=(u,v)$ such that $u$ is explored and $v$ is unexplored. 
A path is \emph{internally explored} if each of its internal vertices is explored.
Given two learned vertices $x$ and $y$, we set the distance $d(x,y)$ to be the length of a shortest internally explored path linking $x$ with $y$. 
In particular, the distance may decrease during execution.
\begin{definition}[Kalyanasundaram and Pruhs \cite{pruhs94}]
Given some $\delta>0$, we say that a boundary edge $e=(u,v)$ is \emph{$\delta$-blocked} if there is another boundary edge $e'=(u',v')$ such \linebreak that ${w(e')<w(e})$ and $\dist (u, v') \leq (1+\delta)w(e)$. 
\end{definition}

The rough idea of $\blocking$ is to perform a depth-first-traversal while ignoring all blocked edges. Whenever a previously blocked edge turns unblocked, the agent moves to and explores one such edge, and initiates a DFS-traversal from its new position. 
$\blocking$ is formally specified in~\cref{alg:blocking}. 
It is executed on an undirected, weighted, connected, and initially unexplored graph $G=(V,E,w)$ and takes as input a vertex $v$ of $G$, denoting the current position of the agent. The algorithm follows a recursive DFS-like structure and the input of the initial invocation is the start vertex.

\begin{algorithm}[H]
\caption{$\blocking_\delta(v)$ \cite{pruhs94, megow} }\label{alg:blocking}
\While{there is a boundary edge $e=(y,x)$ that is not $\delta$-blocked and such that $y=v$ or~$e$ was previously blocked by some edge $(u,v)$}{
traverse a shortest internally explored path from $v$ to $y$\\
traverse $e$\\
$\blocking_\delta(x)$\\
traverse a shortest internally explored path from $x$ to $v$}
\end{algorithm}

Observe that the algorithm is correct, i.e., every vertex is explored: Assume, for the sake of contradiction, that some vertex remains unexplored when the algorithm terminates, i.e., there are still boundary edges. 
Let $e=(u,v)$ be a boundary edge of minimum weight. Then, $e$ is not $\delta$-blocked. Therefore, either the exploration of $u$ should have triggered the exploration of~$v$, or $v$ should have been explored at the last point in time the edge turned unblocked.

\subsection{Key properties of $\blocking$}

Throughout the remainder of Section 2, let $G=(V,E,w)$ be a graph, $n=|V|$ be its number of vertices, $v$ the given start vertex of $G$, and $\delta=\delta(n)>0$. We analyze the performance of $\blocking_\delta$ on $G$, i.e., we estimate its total cost $\cost(G,v, \delta)$. For this, let $B$ be the set of boundary edges taken by $\blocking_\delta$, i.e., the edges traversed during the execution of line~3.

Note that the total cost of the offline optimum is bounded from below by the weight of a minimum spanning tree $w(\mst)$ and from above by~$2w(\mst)$. 
That is, to show that $\blocking_\delta$ is $\rho$-competitive, it suffices to show $\cost(G,v,\delta)\leq \rho \cdot w(\mst)$. 

\begin{observation}[Megow et al. \cite{megow}]\label{obs:boundP}
We have $\cost (G,v, \delta) \leq 2 (\delta +2) w(B)$.
\end{observation}

\begin{proof} 
We charge all cost incurred in lines 2,3, and 5 to the corresponding boundary edge~$e \in B$. 
Note that the cost in line 2 is at most $(1+\delta)w(e)$, because either we have $y=v$ such that $d_G(v,y)=0$, or $e$ was blocked by an edge $(u,v)$, which implies $d_G(y,v) \leq (1+\delta) w(e)$.  
The cost in line~3 is $w(e)$ and the cost in line 5 is at most the sum of the cost in lines~2 and~3. Therefore, each edge $e$ in $B$ is charged at most $2(\delta+2)w(e)$.
\end{proof}

\newcommand{\mstb}{\mathrm{MST}_B}

In our subsequent analysis, we will frequently use a minimum spanning tree with a particular property. For this, in what follows, let $\mstb$ be a minimum spanning tree of~$G$ that maximizes the number of edges in $\mstb \cap B$.
As pointed out in~\cite{megow}, cycles in $B\cup \mstb$ are long relative to the weight of the edges they contain. 
Specifically, the following holds.
\footnote{The assertion of Lemma~\ref{lem:blockingcycles} implies Claim 1 in \cite{megow}, which only concerns edges not in the minimum spanning tree. 
However, there is a subtle flaw in the proof of Claim 1 in \cite{megow}. 
In fact, in that proof, it is not clear that when we replace the edge~$e'$ with an edge of the fixed MST, we again obtain a minimum spanning tree. In any case, the argument above rectifies this. }

\begin{lemma}\label{lem:blockingcycles}
Let $C$ be a cycle in $B \cup \mstb$ and $e$ be an edge of $C$. Then, 
\begin{equation*}
 w(C \setminus \{e\})> (1+\delta) w(e).
 \end{equation*}
\end{lemma}

\begin{proof}
It suffices to show the assertion for an edge of maximum weight in $C$. 
We first show that this edge must be in $B$, i.e., $\argmax \{w(e)\colon e \in C \} \subset B$: 

Assume otherwise and let $e=(u,v) \in \argmax \{w(e)\colon e \in C \} \cap (\mstb \setminus B)$. 
Removing~$e$ from $\mstb$ separates $\mstb$ into two connected components. 
In particular, $u$ and $v$ are in different components. 
Start walking in $C \setminus \{e\}$ from $u$ to $v$ and let $e'$ be the first edge that leads from $u$'s connected component in $\mstb \setminus \{e\}$ to $v$'s connected component. Then, $e' \in B \setminus \mstb$ and by maximality of $e$, we have $w(e')\leq w(e)$. 
Therefore, replacing~$e$ by~$e'$ in $\mstb$ gives another spanning tree of weight at most $w(\mstb)$. This new spanning tree has one more edge in common with $B$ than $\mstb$. This contradicts the choice of $\mstb$, so that we can assume from now on  $\argmax \{w(e)\colon e \in C \} \subseteq B$, i.e., every edge in $\argmax \{w(e)\colon e \in C \}$ is \emph{charged}, i.e., the edge is traversed in some exectution of line 3 of the algorithm.

Let $e=(u,v)$ be the edge in $\argmax \{w(e)\colon e \in C \}$ that is charged last. At the time $e$ is traversed, it is a boundary edge, so that $u$ is explored but $v$ is not yet explored. Start walking in $C \setminus \{e\}$ from $u$ to $v$ and let $e'=(u',v')$ be the first edge leading from an explored vertex $u'$ to an unexplored vertex $v'$, i.e., $e'$ is another boundary edge in $C$ (cf.~\cref{fig:blockingcycles}).

\setlength{\belowcaptionskip}{-10pt}
\begin{figure}
\begin{center}
\begin{tikzpicture}[scale=0.9]
\node (u) at (0,0) [circle, draw, inner sep=0, minimum size=12pt] {$u$};
\node (v) at (1,0) [circle, draw, inner sep=0, minimum size=12pt, green!50!black, thick] {\color{black}$v$};
\draw[-, thick, blue] (u) to node [below] {\color{black}$e$} (v);
\node (u1) at (0,2.5) [circle, draw, inner sep=0, minimum size=12pt] {$u'$};
\node (v1) at (1,2.5) [circle, draw, inner sep=0, minimum size=12pt, green!50!black, thick] {\color{black}$v'$};
\draw[-, thick, blue] (u1) to node [below] {\color{black}$e'$} (v1);
\node (A) at (-1, 0.5)  [circle, draw, inner sep=0, minimum size=12pt] {};
\node (B) at (-1, 2)  [circle, draw, inner sep=0, minimum size=12pt] {};
\draw[-] (u) to (A);
\draw[-] (B) to (u1);
\draw[dashed, bend left=20] (A) to (B);
\draw[dashed, bend left=70] (v1) to (v);
\end{tikzpicture}
\end{center}
\caption{Illustration of \cref{lem:blockingcycles}: The black vertices are explored and the green vertices ($v$~and~$v'$) are unexplored. The blue edges ($e$ and $e'$) are boundary edges.}\label{fig:blockingcycles}
\end{figure}
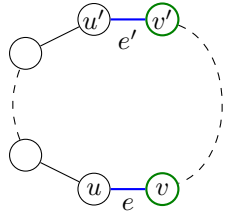

Next, we show that $w(e')<w(e)$: Assume otherwise. By maximality of $e$, this means $w(e')=w(e)$ so that $e' \in \argmax \{w(e)\colon e \in C \}$. But then, we also have $e' \in B$. This contradicts the fact that $e$ is the edge in $\argmax \{w(e)\colon e \in C \}$ that is charged last.

Summing up, we have shown the following facts: Upon exploration of $e=(u,v)$, there is another boundary edge $e'=(u',v')$ in $C$ with $w(e')<w(e)$. Since $e$ is not blocked, this implies 
\begin{equation*}
w(C \setminus \{e\}) \geq d(u,v') > (1 +\delta) w(e).  \qedhere
\end{equation*}
\end{proof}

\subsection{Connection to spanners}\label{sec:connectionspanner}

Next, we investigate how the performance of $\blocking_\delta$ is related to graph spanners. 
For this, note that \cref{lem:blockingcycles} can be reformulated as follows.

\begin{lemma}\label{lem:blockingspanner}
No proper subgraph of $B \cup \mstb$ is a $(1+\delta)$-spanner of $B \cup \mstb$.
\end{lemma}

The lemma relates spanners to the behavior of $\blocking_\delta$.
However, we need to take note that the lemma applies to $B \cup \mstb$ rather than the original graph~$G$.
A \emph{monotone graph class} is a class of graphs $\mathcal{G}$ closed under taking subgraphs, i.e., if $G \in \mathcal{G}$ and $H$ is a subgraph of $G$, then also $H \in \mathcal{G}$. Given a graph $G$, we define $\optspan (G)$ as the minimum lightness of a $(1+\delta)$-spanner of $G$. Moreover, we set  
$\optspan ( \mathcal{G}):=\sup \{\optspan (G) : G \in \mathcal{G} \}$ to be the supremum over all graphs in~$\mathcal{G}$.

\begin{theorem}\label{mainthm}
For every  monotone graph class $\mathcal{G}$ and every $\delta=\delta(n)>0$, the algorithm $\blocking_\delta$ is $\left( 2(\delta+2)\cdot \optspan (\mathcal{G}) \right)$-competitive on $\mathcal{G}$.
\end{theorem}

\begin{proof}
Let $G \in \mathcal{G}$. We have
\begin{equation}\label{eq:mainthm1}
\cost(G,v, \delta) \overset{\text{Obs \ref{obs:boundP}}}{\leq} 2(\delta+2) w(B)
\leq 2(\delta+2) w(B \cup \mstb).
\end{equation}
Since $B \cup \mstb$ is a subgraph of $G$, we have $B \cup \mstb \in \mathcal{G}$. By \cref{lem:blockingspanner}, the only $(1+\delta)$-spanner of $B\cup \mstb$ is $B\cup \mstb$ itself. Therefore, 
\begin{equation}\label{eq:mainthm2}
 w(B \cup \mstb) \leq  \optspan (B\cup \mstb) \cdot w(\mstb) \leq  \optspan (\mathcal{G}) \cdot w(\mstb).
 \end{equation} 
Combined, we obtain
\begin{equation*}
\cost(G,v, \delta) \overset{\eqref{eq:mainthm1}}{\leq} 2(\delta+2) w(B \cup \mstb) \overset{\eqref{eq:mainthm2}}{\leq} 2(\delta+2) \cdot \optspan (\mathcal{G}) \cdot w(\mstb). \qedhere
\end{equation*} 
\end{proof}

The theorem puts us in a position to leverage results on the lightness of spanners in order to draw conclusions regarding the competitive ratio of $\blocking_\delta$. 
For example, it has been shown that 
every planar graph contains a $(1+\delta)$-spanner of lightness at most $1+\frac{2}{\delta}$ \cite{althoefer}. 
Feeding this into \cref{mainthm}, we conclude that $\blocking_\delta$ is $2(\delta+2)(1+2/\delta)$-competitive on planar graphs. 
This agrees with the bound proven in~\cite{pruhs94}.
However, more generally, bounded genus graphs have light spanners. In fact, in Section~\ref{sec:genaral:to:bd:genus}, we show that every graph of genus at most $g$ contains a $(1+\delta)$-spanner of lightness at most $\left(1+\frac{2}{\delta} \right) \bigl( 1 + \frac{2g}{1+\delta}\bigl)$ (\cref{thm:genusspanner}). From this, we obtain the following.

\begin{corollary}\label{cor:blocking:delta:genus}
$\blocking_\delta$ is $2(\delta+2)\bigl(1+\frac{2}{\delta} \bigl) \bigl(1+\frac{2g}{1+\delta} \bigl)$-competitive on graphs of genus at most $g$.
\end{corollary}

Even more generally, it is known that $H$-minor-free graphs have light spanners~\cite{minorspanner}. Specifically, every $H$-minor-free graph contains a  $(1+\delta)$-spanner of lightness  $O\left( \frac{\sigma_H}{\delta^3} \log \left(\frac{1}{\delta} \right)\right)$ where $\sigma_H=|V(H)|\sqrt{\log |V(H)|}$. This yields a constant competitive ratio for $\blocking_\delta$ on $H$-minor-free graphs as follows.

 \begin{corollary}\label{cor:blocking:delta:minor}
 $\blocking_\delta$ is $2(\delta + 2) \cdot O\left( \frac{\sigma_H}{\delta^3} \log \left( \frac{1}{\delta} \right) \right)$-competitive on $H$-minor-free graphs where $\sigma_H=|V(H)|\sqrt{\log |V(H)|}$.
 \end{corollary}

There are also strong bounds for general graphs.  Given a graph $G$ with $n$ vertices and an integer $k\geq 1$ and $\epsilon \in (0,1)$, $G$ contains a $(2k-1)(1+\epsilon)$-spanner of lightness $O_\epsilon\left(n^{1/k} \right)$~\cite{chechik}, where the notation $O_\epsilon$ indicates that the constant factor hidden in the $O$-notation depends on $\epsilon$. This gives us the following.

\begin{corollary}\label{cor:blocking:delta:general}
Given an integer $k=k(n)\geq 1$ and $\epsilon \in (0,1)$, $\blocking_{(2k-1)(1+\epsilon)}$ is ${2((2k-1)(1+\epsilon) + 2)} \cdot O_\epsilon\left(n^{1/k}  \right)$-competitive on every graph.
\end{corollary}

In particular, by suitably choosing $\delta$, we obtain the following.

\begin{restatable}{corollary}{blockingcorollary}\label{cor:blockinggenus}
\begin{enumerate}[a)]
\item $\blocking_2$ is $16(1+\frac{2}{3}g)$-competitive on graphs of genus at most $g$.\label{thmpart:genus}
\item For every constant $\delta>0$ and every graph $H$, $\blocking_\delta$ is constant-competitive on $H$-minor-free graphs.\label{thmpart:minor}
\item $\blocking_{\log (n)}$ is $O(\log (n))$-competitive on every graph.\label{thmpart:general}
\end{enumerate}
\end{restatable}

\begin{proof}
Part \ref{thmpart:genus}) follows from Corollary~\ref{cor:blocking:delta:genus} by setting~$\delta=2$
and part \ref{thmpart:minor}) follows from Corollary~\ref{cor:blocking:delta:minor}.

For part \ref{thmpart:general}), note that one can choose an integer $k\in \Theta(\log (n))$ and $\epsilon \in (0.5,1)$ such that $(2k-1)(1+\epsilon)=\log (n)$. For the competitive ratio from Corollary~\ref{cor:blocking:delta:general}, we then obtain
\begin{equation*}
2((2k-1)(1+\epsilon) + 2) \cdot O_\epsilon(n^{1/k})= O( \log (n) \cdot n^{1/\log (n)})=O(\log (n)).\qedhere
\end{equation*} 
\end{proof}

For the case of planar graphs, part \ref{thmpart:genus}) matches the best known bounds on planar \hbox{graphs~\cite{ pruhs94, megow}}. For general surfaces, it slightly improves on the best known bound of $16(1+2g)$ on bounded genus graphs \cite{megow}. Part~\ref{thmpart:minor}) is the first constant bound on minor-free graphs, and part~\ref{thmpart:general}) is the first $O(\log (n))$ bound for $\blocking$.

\subsection{Lower bounds for $\blocking$}\label{sec:lowerbounds}

Next, we investigate lower bounds for $\blocking$ when~$\delta$ is allowed to depend on the input size.
In \cite{megow}, it was shown that the competitive ratio of $\blocking_\delta$ on general graphs is at least $\Omega(n^{1/(\delta+4)})$ when $\delta$ is a constant. 
We begin by observing that this can be generalized to non-constant $\delta$ that are not too large.

\begin{restatable}{observation}{obslowerbound}\label{obs:lowerconstruction}
Suppose $\delta=\delta(n)>0$ such that $\delta^{2\delta+8}=o(n)$. Then, the competitive ratio of $\blocking_\delta$ is at least $\Omega(\delta \cdot n^{1/(\delta+4)})$. 
\end{restatable}

\begin{proof}
A careful investigation of the construction in \cite[Theorem 3]{megow} with the following slight change gives the desired bound: The ``blocking paths'' and ``center paths'' are replaced by $d+1$ edges of weight~1, followed by a suitable number of heavier edges of weight $d/(1+\delta)$. If $d\geq \delta^2$, this ensures that the number of vertices in a gadget only depends on $d$ but not on~$\delta$. 
\end{proof}

Note that $2\delta+8 \leq \log (n)/\log(\log(n))$ implies that
\begin{align*}
\delta^{2 \delta + 8} &\leq \left( \frac{\log (n)}{\log(\log(n))} \right)^{\frac{\log (n)}{\log(\log(n))}}
=  \left( \frac{1}{\log(\log(n))} \right)^{\frac{\log (n)}{\log(\log(n))}}\cdot e^{\log ( \log (n))\frac{\log (n)}{\log(\log(n))}}\\
&= \left( \frac{1}{\log(\log(n))} \right)^{\frac{\log (n)}{\log(\log(n))}} \cdot n
=o(n),
\end{align*}
i.e., the prerequisites of \cref{obs:lowerconstruction} are fulfilled. Moreover, $\Omega(\delta n^{1/(\delta+4)})\geq\Omega(\log (n))$ for every $\delta=\delta(n)$. Therefore, \cref{obs:lowerconstruction} shows that $\blocking_\delta$ has competitive ratio in $\Omega(\log (n))$  whenever $\delta = o(\log (n)/\log \log (n))$.  In particular, this shows the first part of \cref{thm:blockinglower}\ref{thmpart:specific:delta:lowerbound}.

Next, we give another lower bound which shows that the parameter $\delta$ cannot be chosen too large either (the second part of \cref{thm:blockinglower}\ref{thmpart:specific:delta:lowerbound}).

\begin{restatable}{lemma}{lemmalowerbound}\label{lem:deltalower}
Suppose $\delta=\delta(n) \in (0,\frac{n-4}{4})$. The competitive ratio of $\blocking_\delta$ is at least~$\Omega(\delta)$, even on trees.
\end{restatable}

\begin{proof}
Given a
positive
integer $k$, 
we construct a graph $G$ as follows (see \cref{fig:lowerconstruction}): We begin with a path of $2k$ vertices and edges of weight 1. The first vertex $v$ of this path is the start vertex. To the first $k$ vertices, we attach leaves by edges $e_1, \dots, e_k$ of weight 1 and call these \emph{light edges}. To the last $k$ vertices, we attach leaves by edges $e_1', \dots, e_k'$ of weight $h\coloneqq\frac{k+1}{\delta+1}$ and call these \emph{heavy edges}. Note that $n=4k$ and $h>1$.

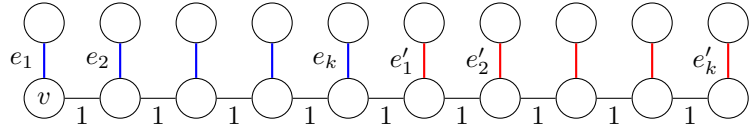
\begin{figure}
\begin{center}
\begin{tikzpicture}[scale=1]
\node (v1) at (1,0) [circle, draw, inner sep=0, minimum size=15pt] {$v$};
\foreach \x in {2,...,10}{
\node (v\x) at (\x,0) [circle, draw,inner sep=0, minimum size=15pt] {};
}
\foreach \x in {1,...,9}{
\draw[-] (v\the\numexpr \x + 1 \relax) to  node [below] {1} (v\x);
}
\foreach \x in {1,2}{
\node (A) at (\x,1) [circle, draw, inner sep=0, minimum size=15pt] {};
\draw[-, blue, thick] (A) to node [left] {\color{black}$e_{\x}$}(v\x);
}
\foreach \x in {3,4}{
\node (A) at (\x,1) [circle, draw, inner sep=0, minimum size=15pt] {};
\draw[-, blue, thick] (A) to (v\x);
}
\node (A) at (5,1) [circle, draw, inner sep=0, minimum size=15pt] {};
\draw[-, blue, thick] (A) to node [left] {\color{black}$e_{k}$}(v5);
\foreach \x in {6,7}{
\node (A) at (\x,1) [circle, draw, inner sep=0, minimum size=15pt] {};
\draw[-, red, thick] (A) to node [left] {\color{black}$e_\the\numexpr \x -5 \relax '$}(v\x);
}
\foreach \x in {8,9}{
\node (A) at (\x,1) [circle, draw, inner sep=0, minimum size=15pt] {};
\draw[-, red, thick] (A) to (v\x);
}
\node (A) at (10,1) [circle, draw, inner sep=0, minimum size=15pt] {};
\draw[-, red, thick] (A) to node [left] {\color{black}$e_k '$}(v10);
\end{tikzpicture}
\end{center}
\caption{ Illustration of the lower bound construction for $\blocking_\delta$ (\cref{lem:deltalower}). The light edges (depicted in blue) are of weight~1 and the heavy edges (depicted in red) are of weight $\frac{k+1}{\delta + 1}$.}\label{fig:lowerconstruction}
\end{figure}

Next, observe that $\blocking_\delta$ explores the graph in the following way:
We can adversarially assume
that $\blocking_\delta$ begins by exploring the path of $n/2$ vertices (and none of the light edges).
\footnote{Upon exploring one of the first $k$ path vertices, the agent cannot distinguish  the incident light edge from the next edge of the path because they both have weight 1 and lead to an unexplored vertex. Therefore, it is not specified in the algorithm which edge should be traversed next. As an adversary, we can simply choose after every step that the edge traversed by the agent is the path edge.} Then, the heavy edge $e_i'=(u_i',v_i')$ is $\delta$-blocked by the light edges $e_i \dots, e_k$, because $w(e_i')>1$ and the distance from $u_i'$ to the
unexplored
end vertex of $e_i$ is
\begin{equation*}
	k+1=(1+\delta) \frac{k+1}{\delta+1} \leq (1+\delta) w(e_i').
\end{equation*}
But $e_i'$ is not blocked by the light edges $e_1, \dots, e_{i-1}$, because the distance from $u_i'$ to the end point of $e_{i-1}$ is
\begin{equation*}
	k+2=(1+\delta) \frac{k+2}{\delta+1} > (1+\delta) w(e_i').
\end{equation*}
 Therefore, the agent explores, for $i=1, \dots, k$, the $(k-i)$-th light edge and then the $(k-i)$-th heavy edge, forcing it to travel a distance of more than $k$ at least $k$ times.
 Hence,
\begin{equation*}
\cost(G,v, \delta)\geq k^2.
\end{equation*}
Since $G$ is a tree, the optimal tour has cost 
\begin{equation*}
2w(G)=2 \left(2k-1+k + \frac{k^2}{(1+\delta)}\right).
\end{equation*}
Therefore, the competitive ratio of $\blocking_\delta$ is at least
\begin{equation*}
\frac{\cost (G,v, \delta)}{2 w(G)}\geq \frac{k^2}{2\left(\frac{k^2}{1+\delta}+3k-1\right)}\to \frac{\delta+1}{2} ~~ (k \to \infty).
\end{equation*}
Since $k=n/4$ so that $n \to \infty$ implies $k\to \infty$, this proves the lower bound of $\Omega(\delta)$.
\end{proof}

To conclude our lower bound arguments for $\blocking_\delta$, observe that, for $\delta\geq \frac{n-4}{4}$, the behavior of $\blocking_\delta$ closely resembles the behavior of the algorithm hDFS \cite{megow}. In fact, it is not difficult to check that, on the lower bound construction for hDFS given \hbox{in~\cite[Theorem 5]{megow}}, after proceeding to edges of weight more than 16, $\blocking_\delta$ takes the exact same route as hDFS, if $\delta\geq \frac{n-4}{4}$. Therefore, we obtain the following.

\begin{observation}\label{obs:hdfs}
For $\delta \geq  \frac{n-4}{4}$, the competitive ratio of $\blocking_\delta$ is at least $\Omega(\log (n))$.
\end{observation}

We can now combine the lower bound constructions from this section to 
 prove \hbox{\cref{thm:blockinglower}}.

\begin{proof}[Proof of \cref{thm:blockinglower}]
We begin with proving part b). In \cref{obs:lowerconstruction}, we have seen that the competitive ratio of $\blocking_\delta$ is at least $\Omega(\log (n))$ if $\delta \in o(\log (n)/\log \log (n))$. By \cref{lem:deltalower}, we obtain the same lower bound for every $\delta$ in the range from $\Omega(\log (n))$ to $(n-4)/4$, and by \cref{obs:hdfs}, we obtain the lower bound for $\delta$ at least $(n-4)/4$. Therefore, this proves the assertion of  \cref{thm:blockinglower}b. For part a), note that part b) implies that a competitive ratio of $o(\log (n))$ is only possible for $\delta$ in the range from $\Omega(\log (n)/\log \log (n))$ to $o(\log (n))$. Using \cref{obs:lowerconstruction} in this range implies the assertion of \cref{thm:blockinglower}a.
\end{proof}

\section{Graph spanners in bounded genus graphs}

In this section, we prove \cref{thm:genusspanner} about the existence of light spanners in bounded genus graphs. For this, we begin by introducing the greedy spanner.

\subsection{The greedy spanner}

The \emph{greedy $(1+\epsilon)$-spanner} was suggested by Althöfer et al.~\cite{althoefer} and is formally defined as the output of \cref{alg:greedyspanner}. After ordering the edges by weight, it iteratively adds edges if they are short in comparison to the distance of their endpoints in the graph constructed so far.
\begin{algorithm}
\caption{\textsc{GreedySpanner}($G=(V,E, w)$, $\epsilon$)}\label{alg:greedyspanner}
sort $E=\{e_1, \dots, e_m \}$ such that $w(e_1)\leq w(e_2) \leq \dots \leq w(e_m)$\\
$H\gets (V,\emptyset)$\\
\For{ $i \leftarrow 1 , \dots , m$}{
\If{$\dist_H(u_i,v_i)>(1+\epsilon)w(e_i)$, where $e_i=(u_i,v_i)$}{$H\gets H\cup \{e_i\}$}
}
\Return $H$
\end{algorithm}

Note that the resulting graph $H$ is indeed a $(1+\epsilon)$-spanner of $G$.
The output of the algorithm actually depends on the chosen order of the edges. 
In particular, when edge weights appear multiple times, there may be several possible outputs. However, this will not be important in our context. 
When we refer to \emph{the} greedy spanner, we mean that we arbitrarily fix some output of the algorithm.

The greedy spanner fulfills the following two key properties: First, the algorithm implicitly executes Kruskal's algorithm for finding a minimum spanning tree, i.e., it adds all edges to~$H$ that Kruskal's algorithm adds. With this, we obtain the following.
\begin{observation}\label{greedyMST}
The greedy spanner contains all edges of some minimum spanning tree of the input graph.
\end{observation}
The second key property, in fact, resembles the property of $\blocking_\delta$ in Lemma~\ref{lem:blockingcycles}.
\begin{observation}[Althöfer et al.~\cite{althoefer}]\label{obs:longcycle}
For every cycle $C$ in the greedy spanner $H$ and every edge $e$ of $C$, we have $w(C\setminus \{e\})>(1+\epsilon) w(e)$. In other words, no proper subgraph of $H$ is a $(1+\epsilon)$-spanner of $H$.
\end{observation}
\begin{proof}
Let $C$ be a cycle in the greedy spanner. Let $e=(u,v)$ be the edge in $C$ that is added last. At the time it is added, we have $(1+\epsilon)w(e)<d_H(u,v)\leq w(C\setminus \{e\})$ by definition of the algorithm. Since all other edges in $C$ have lower or equal weight than $e$, the property is fulfilled for them as well.
\end{proof}

\subsection{Spanners in planar graphs}\label{sec:planar}

Before investigating spanners in bounded genus graphs, we illustrate the technique for the special case of planar graphs, giving an alternate proof of the following result. 
\begin{theorem}[Althöfer et al.~\cite{althoefer}]\label{thm:planarspanner}
For every planar graph $G$ and $\epsilon >0$, the greedy $(1+\epsilon)$-spanner of $G$ has lightness at most $1+\frac{2}{\epsilon}$.
\end{theorem}

Our proof uses similar ideas as in \cite[Theorem 1]{megow} and is based on the following main idea: 
Fix an embedding of the greedy spanner in the plane and, in a suitable way, partition the greedy spanner into facial cycles, i.e., cycles that form the boundary of a face. 
Then use the fact that none of these cycles are short {(cf. \cref{obs:longcycle})}.

\begin{restatable}{lemma}{lemfacialcycles}\label{lem:facialcycles}
Let $G$ be a planar graph, $H$ be the greedy $(1+\epsilon)$-spanner of $G$ and $\mst$ be a minimum spanning tree of $H$. Fix an embedding of $H$ in the plane. Then, we can associate with every edge $e\in H \setminus \mst$ a facial cycle $C_e$ containing $e$, so that $C_e\neq C_{e'}$ for $e \neq e'$.
\end{restatable}

\begin{proof}
We show that it is possible to iteratively choose an edge $e$ in $H \setminus \mst$ that closes a facial cycle $C_e$ together with the edges of $\mst$ and the edges chosen in previous iterations:
We define a partial order on the edges in $H \setminus \mst$. Every edge $e \in H \setminus \mst$ closes a cycle~$C$ together with the edges of $\mst$. We let another edge $e'$ precede $e$ in the partial order if it lies on the inside of this cycle in the considered embedding (see \cref{fig:partialorder})\footnote{Here, the \emph{inside} of the cycle is, out of the two regions bounded by the cycle, the one that does not contain the outer face of $G$.}.
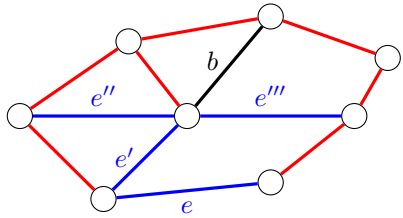
\begin{figure}
\begin{center}
\begin{tikzpicture}[scale=1.1]
\node (A) at (0,0) [circle, draw] {};
\node (B) at (2,0.2) [circle, draw] {};
\node (Ct) at (3.4,1.7) [circle, draw] {};
\node (C) at (3,1) [circle, draw] {};
\node (D) at (2,2.2) [circle, draw] {};
\node (E) at (0.3,1.9) [circle, draw] {};
\node (F) at (-1,1) [circle, draw] {};
\node (G) at (1,1) [circle, draw] {};
\draw[-, very thick, blue] (A) to node [below] {$e$} (B);
\draw[-, very thick, red] (B) to (C);
\draw[-, very thick, red] (C) to (Ct);
\draw[-, very thick, red] (Ct) to (D);
\draw[-, very thick, red] (D) to (E);
\draw[-, very thick, red] (E) to (G);
\draw[-, very thick, red] (E) to (F);
\draw[-, very thick, red] (F) to (A);
\draw[-, very thick, blue] (A) to node [left] {$e'$} (G);
\draw[-, very thick, blue] (F) to node [above] {$e''$} (G);
\draw[-, very thick] (D) to node [left, yshift=2pt] {$b$} (G);
\draw[-, very thick, blue] (C) to node [above] {$e'''$} (G);
\end{tikzpicture}
\end{center}
\caption{Illustration of the partial order in \cref{lem:facialcycles}: The red edges denote $\mst$, the black edge $b$ has already been chosen in a previous iteration and the blue edges $e, e', e'', e'''$ have not yet been assigned a facial cycle. In this example, $e'$, $e''$ and $e'''$ precede $e$. Moreover, $e''$ precedes $e'$. The edges $e''$ and $e'''$ are minimal. Note that $e'''$ was not minimal before the black edge was chosen. If $e'''$ is chosen in this step, the facial cycle assigned to $e'''$ consists of the black edge $b$, $e'''$, and two red edges.}\label{fig:partialorder}
\end{figure}
In each iteration, we can choose an edge which is minimal in this partial order amongst the edges to which no cycle has yet been assigned.
Note that, in this construction, no two edges are assigned the same cycle.
\end{proof}

Next, we illustrate how this can be combined with the fact that the greedy spanner does not contain short cycles (cf. \cref{obs:longcycle}).

\begin{restatable}{lemma}{lemspannercomb}\label{lem:spannercomb}
Let $G$ be a graph and $H$ be the greedy $(1+\epsilon)$-spanner of $G$. Let~$D$ be a subgraph of~$G$ such that we can associate with
every edge $e \in H \setminus D$ a cycle $C_e$ of $H$ containing~$e$, 
with the property that 
$\sum_{e \in H \setminus D} w(C_e) \leq 2w(H)$. Then,
$
w(H) \leq \left(1+\frac{2}{\epsilon}\right)w(D).
$
\end{restatable}

\begin{proof}
We obtain
\begin{align*}
w(H \setminus D) &= \sum\limits_{e \in H \setminus D} w(e)
\overset{\text{Obs \ref{obs:longcycle}}}{<} \frac{1}{1+\epsilon}\sum\limits_{e \in H \setminus D} w(C_e \setminus \{e\})\\
&=\frac{1}{1+\epsilon} \left( \sum\limits_{e \in H \setminus D} w(C_e) - \sum\limits_{e \in H \setminus D} w(e)\right)\\
&\leq \frac{1}{1+\epsilon} \left( 2 w(H) - w(H \setminus D ) \right)
= \frac{1}{1+\epsilon} \left( 2 w(D) + w(H \setminus D ) \right).
\end{align*}
Rearranging yields
\begin{equation}\label{eq:HminusMST}
w(H \setminus D ) \leq \frac{2}{\epsilon} \cdot w(D)
\end{equation}
and thus
\begin{equation*}
w(H) =w(H \setminus D) + w(D) 
\overset{\eqref{eq:HminusMST}}{\leq} \left( 1+\frac{2}{\epsilon} \right) w(D).\qedhere
\end{equation*}
\end{proof}

Next, we show how this implies \cref{thm:planarspanner}.
\begin{proof}[Proof of \cref{thm:planarspanner}]
Let $G$ be a planar graph, let $H$ be the greedy $(1+\epsilon)$-spanner of $G$, and let $\mst$ denote a minimum spanning tree of $H$. By \cref{greedyMST}, $\mst$ is also a minimum spanning tree of $G$, so that it suffices to show ${w(H)\leq \left( 1+\frac{2}{\epsilon}\right) w(\mst)}$.
Since $G$ is planar, its subgraph $H$ is planar as well. Let us fix an embedding of $H$ on the plane such that no two edges cross. By \cref{lem:facialcycles}, there is a facial cycle $C_e$ for every edge $e \in H \setminus \mst$ such that $C_e \neq C_{e'}$ for $e \neq e'$. As every edge of $H$ is contained in at most two facial cycles, we have $\sum_{e \in H \setminus \mst} w(C_e) \leq 2 w(H)$. Therefore, we can apply \cref{lem:spannercomb} with $D=\mst$ and obtain
$
w(H) \leq \left( 1+\frac{2}{\epsilon} \right) w(\mst). \qedhere
$
\end{proof}

\subsection{Generalization to bounded genus graphs}\label{sec:genaral:to:bd:genus}

The \emph{genus} of a graph $G$ is the smallest integer $g$ such that $G$ can be embedded on an orientable surface of genus $g$. 
In this subsection, we study light spanners for the class of bounded genus graphs and prove \cref{thm:genusspanner}. 
We begin by recalling the theorem.

\newtheorem{restatedtheorem}{Theorem}
\setcounter{restatedtheorem}{\value{restatecounter}}

\begin{restatedtheorem}[\text{restated}]
For every $\epsilon>0$, the greedy $(1+\epsilon)$-spanner of a graph of genus $g$ has lightness at most $\bigl( 1+\frac{2}{\epsilon} \bigl) \bigl(1+\frac{2g}{1+\epsilon} \bigl)$.
\end{restatedtheorem}

Our proof is based on similar arguments as in  \cite[Theorem 2]{megow} and the main idea is roughly as follows: Given an embedding of the greedy spanner on a surface of genus $g$, first cut the surface along several edges such that we obtain a disk. 
Then, we can proceed along similar lines as for \cref{thm:planarspanner}. In this work, we estimate more carefully the weight of the edges along which we cut so that we obtain a slightly improved bound than in \cite{megow}.
We will use the following topological lemma for the first step.
\begin{lemma}\label{lem:constructionD}
 Let $G$ be an unweighted connected graph of genus (exactly) $g\geq 1$.
Fix an embedding of $G$ on an orientable surface of genus $g$ and let $T$ be a spanning tree of $G$. Then, there is a subgraph $D$ of $G$ with $T \subset D$ and $|E(D) \setminus E(T)|\leq 2g$ such that, in the inherited embedding of $D$, there is only a single face and the edges in $D$ bound a topological disk.\footnote{A topological disk is a surface homeomorphic to a 2-dimensional disk. Intuitively, a topological disk is a continuous deformation of a 2-dimensional disk.}
\end{lemma}
\begin{proof}
It is a standard fact from topology that, on a surface of genus $g$, one can embed precisely $2g$ closed curves that are non-separating, i.e., it is possible to draw $2g$ cycles on the surface such that cutting along all of them does not disconnect the surface. Every collection of $2g$ curves that are non-separating bounds a topological disk (see \cref{fig:doubletorus}).\footnote{This can be proven as follows: The Euler characteristic of a surface of genus $g$ is $2-2g$ \cite[Section 2.2]{hatcher} and cutting along a non-separating closed curve increases the Euler characteristic by 1.}
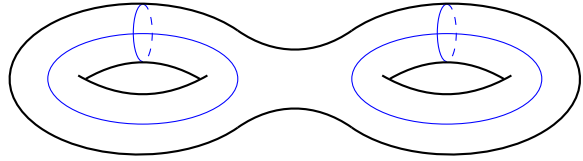
\begin{figure}
\begin{center}
\begin{tikzpicture}[scale=0.5]
\draw[blue] (-2.5,0.7) arc(90:270:0.25cm and 0.76cm);
\draw[blue, dashed] (-2.5,0.7) arc(90:-90:0.25cm and 0.76cm);
\draw[blue] (5.5,0.7) arc(90:270:0.25cm and 0.76cm);
\draw[blue, dashed] (5.5,0.7) arc(90:-90:0.25cm and 0.76cm);
\draw[blue] (0,-1.28) arc(0:360:2.5cm and 1.2cm);
\draw[blue] (3,-1.28) arc(180:-180:2.5cm and 1.2cm);
\draw[thick] (0,0) arc(40:320:3.4cm and 2cm);
\draw[bend right=35, thick] (0,0) to (3,0);
\draw[bend left=35, thick] (0,-2.57) to (3,-2.57);
\draw[thick] (3,0) arc(140:-140:3.4cm and 2cm);
\draw[-, bend left, thick] (-4,-1.28) to (-1, -1.28);
\draw[-, bend right, thick] (-4.2,-1.19) to (-0.8, -1.19);
\draw[-, bend right, thick] (7,-1.28) to (4, -1.28);
\draw[-, bend left, thick] (7.2,-1.19) to (3.8, -1.19);
\end{tikzpicture}
\end{center}
\caption{surface of genus 2 with 4 non-separating cycles bounding a topological disk}\label{fig:doubletorus}
\end{figure}

We construct the set $D$ greedily as follows (see \cref{fig:construction}): Initially, let $D:=T$. Ignoring all edges in $G\setminus D$, we have only a single face. 
Note that every edge in $G \setminus D$ closes a cycle with $D$. 
If we find an edge which only closes non-separating cycles, i.e, does not separate the surface into two faces, we add it to $D$. After this, the edges of $D$ still only bound a single face. We repeat this step until we cannot find further edges whose addition would separate the surface into multiple faces. 

Since there are at most $2g$ cycles on a surface of genus $g$ that are non-separating, we have $|E(D) \setminus E(T)|\leq 2g$. It is left to show that $D$ bounds a disk. By maximality of $D$, every edge $e \in G \setminus D$ is separating when added to $D$, i.e., in the inherited embedding of $D \cup \{e\}$, the edge $e$ is incident to two faces. In particular, $e$ is incident to two faces in the inherited embedding of every supergraph of $D$. 

Consider again the embedding of the entire graph $G$. 
It is known from topological graph theory that a minimal genus embedding of a connected graph is cellular, i.e., every face of the embedding of $G$ is a topological disk \cite{youngs63} (see \cite[Proposition 3.4.1]{mohar}). 
Since every edge $e \in G \setminus D$ is incident to two distinct faces, its removal merges the two corresponding disks along a connected part of their common boundary, which yields another disk.
Iteratively removing all edges in $G \setminus D$ in this way, we thus obtain a cellular embedding of $D$. 
Since, by construction, $D$ induces only a single face, we obtain that $D$ bounds a topological disk.

For an illustration of the construction, consider \cref{fig:construction}. 
In the example in the left column, the two green edges enclose non-separating cycles, whereas all blue edges close separating cycles. 
In the example in the right column, the half-dotted green edge in $D$ could be replaced by the blue edge between $u$ and $v$.
\end{proof}

\setlength{\belowcaptionskip}{-15pt}
\begin{figure}[h]
\begin{minipage}{0.5\textwidth}
\begin{center}
\begin{tikzpicture}[scale=1.4]
\node [ellipse, minimum width=175pt, minimum height=107pt, draw, thick] at (0,0) {}; 
\draw[-, bend left, thick] (-1,0.05) to (1,0.05);
\draw[-, bend right, thick] (-1.1,0.1) to (1.1,0.1);
\node (A) at (0.8,-0.6) [circle, draw, inner sep=0, minimum size=10pt, thick] {\small{$y$}};
\node (B) at (1.5,0) [circle, draw, inner sep=0, minimum size=10pt] {\small{$z$}};
\draw[-, bend right=10] (A) to (B);
\node (C) at (0.8, 0.5) [circle, draw, inner sep=0, minimum size=10pt] {\small{$v$}};
\draw[-, bend right=10] (B) to (C);
\node (D) at (0, 0.65) [circle, draw, inner sep=0, minimum size=10pt] {};
\node (E) at (0, 1.05) [circle, draw, inner sep=0, minimum size=10pt] {};
\node (Eh) at (0.45, 0.9) [circle, draw, inner sep=0, minimum size=10pt] {\small{$u$}};
\draw[-, bend right=10] (C) to (D);
\draw[-] (D) to (Eh);
\draw[-] (E) to (Eh);
\draw[thick, green!60!black] (0,1.16) arc(39:90:0.16cm and 0.49cm);
\draw[green!60!black, dashed, thick] (-0.13, 1.33) arc(90:270:0.16cm and 0.49cm);
\draw[thick, green!60!black] (-0.13, 0.34) arc(270:323:0.16cm and 0.49cm);
\node (F) at (-0.8, 0.5) [circle, draw] {};
\draw[-, bend right=10]   (D) to (F);
\node (G) at (-1.5, 0) [circle, draw, inner sep=0, minimum size=10pt] {};
\draw[-, bend right=10, thick, green!60!black]   (F) to (G);
\node (H) at (-0.8, -0.6) [circle, draw, inner sep=0, minimum size=10pt] {\small{$x$}};
\node (I) at (0, -0.8) [circle, draw, inner sep=0, minimum size=10pt] {};
\draw[-, bend right=10] (G) to (H);
\draw[-, bend right=10] (H) to (I);
\draw[-, bend right=10] (I) to (A);
\draw[-, bend right=4, blue, thick] (B) to (H);
\draw[-, bend right=10, blue, thick] (A) to (H);
\draw[-, thick, blue] (Eh) to (C);
\end{tikzpicture}
\end{center}
\begin{center}
\begin{tikzpicture}[scale=1.3]
\draw[, thick] (-0.5,0) to (4,0) [postaction={decorate,decoration={markings,
        mark=between positions 0.49 and 0.52 step 0.02 with {\arrow[thick]{>};}
        }}];
\draw[-, thick] (-0.5,0) to (-0.5,2) [postaction={decorate,decoration={markings,
        mark=at position 0.5 with {\arrow[thick]{>};}
        }}];
\draw[-, thick] (4,0) to (4,2) [postaction={decorate,decoration={markings,
        mark=at position 0.5 with {\arrow[thick]{>};}
        }}];
\draw[-, thick] (-0.5,2) to (4,2) [postaction={decorate,decoration={markings,
        mark=between positions 0.49 and 0.52 step 0.02 with {\arrow[thick]{>};}
        }}];
\node (A) at (0.5,0.5) [circle, draw, inner sep=0, minimum size=10pt] {};
\node (u) at (0.5,1)[circle, draw, inner sep=0, minimum size=10pt] {\small{$u$}};
\node (B) at (0.5,1.5)[circle, draw, inner sep=0, minimum size=10pt] {};
\draw[-] (A) to (u);
\draw[-] (B) to (u);
\node (v) at (1, 0.5)[circle, draw, inner sep=0, minimum size=10pt] {\small{$v$}};
\draw[-, thick, blue] (u) to (v);
\draw[-] (A) to (v);
\node (z) at (1.5, 0.5)[circle, draw, inner sep=0, minimum size=10pt] {\small{$z$}};
\draw[-] (v) to (z);
\node (y) at (2, 0.5)[circle, draw, inner sep=0, minimum size=10pt] {\small{$y$}};
\draw[-] (z) to (y);
\node (C) at (2.5, 0.5)[circle, draw, inner sep=0, minimum size=10pt] {};
\draw[-] (y) to (C);
\node (x) at (3, 0.5)[circle, draw, inner sep=0, minimum size=10pt] {\small{$x$}};
\draw[-] (C) to (x);
\draw[-, thick, blue, bend right=35] (y) to (x);
\draw[-, thick, blue, bend right=45] (z) to (x);
\node (D) at (3.5,0.5) [circle, draw, inner sep=0, minimum size=10pt] {};
\node (E) at (0,0.5) [circle, draw, inner sep=0, minimum size=10pt] {};
\draw[-, thick, green!60!black] (-0.5,0.5) to (E);
\draw[-] (A) to (E);
\draw[-] (x) to (D);
\draw[-, thick, green!60!black] (D) to (4,0.5);
\draw[-, thick, green!60!black] (B) to (0.5,2);
\draw[-, thick, green!60!black] (A) to (0.5,0);
\node at (-0.25, 0.25) {\small{$B$}};
\node at (-0.25, 1.5) {\small{$A$}};
\node at (3.25, 0.25) {\small{$C$}};
\node at (3.25, 1.5) {\small{$D$}};
\end{tikzpicture}
\end{center}
\begin{center}
\begin{tikzpicture}[scale=1.3]
\node (A) at (0.5,0.5) [circle, draw, inner sep=0, minimum size=10pt] {};
\node (u) at (0.5,1)[circle, draw, inner sep=0, minimum size=10pt] {\small{$u$}};
\node (B) at (0.5,1.5)[circle, draw, inner sep=0, minimum size=10pt] {};

\draw[-] (A) to (u);
\draw[-] (B) to (u);
\node (v) at (1, 0.5)[circle, draw, inner sep=0, minimum size=10pt] {\small{$v$}};
\draw[-, thick, blue] (u) to (v);
\draw[-] (A) to (v);
\node (z) at (1.5, 0.5)[circle, draw, inner sep=0, minimum size=10pt] {\small{$z$}};
\draw[-] (v) to (z);
\node (y) at (2, 0.5)[circle, draw, inner sep=0, minimum size=10pt] {\small{$y$}};
\draw[-] (z) to (y);
\node (C) at (2.5, 0.5)[circle, draw, inner sep=0, minimum size=10pt] {};
\draw[-] (y) to (C);
\node (x) at (3, 0.5)[circle, draw, inner sep=0, minimum size=10pt] {\small{$x$}};
\draw[-] (C) to (x);
\node (D) at (3.5,0.5) [circle, draw, inner sep=0, minimum size=10pt] {};
\draw[-] (x) to (D);
\node (A2) at (0.5,2) [circle, draw, inner sep=0, minimum size=10pt] {};
\draw[-, thick, green!60!black] (B) to (A2);
\node (v2) at (1, 2)[circle, draw, inner sep=0, minimum size=10pt] {\small{$v$}};
\draw[-] (A2) to (v2);
\node (z2) at (1.5, 2)[circle, draw, inner sep=0, minimum size=10pt] {\small{$z$}};
\draw[-] (v2) to (z2);
\node (y2) at (2, 2)[circle, draw, inner sep=0, minimum size=10pt] {\small{$y$}};
\draw[-] (z2) to (y2);
\node (C2) at (2.5, 2)[circle, draw, inner sep=0, minimum size=10pt] {};
\draw[-] (y2) to (C2);
\node (x2) at (3, 2)[circle, draw, inner sep=0, minimum size=10pt] {\small{$x$}};
\draw[-] (C2) to (x2);
\draw[-, thick, blue, bend right=28] (y2) to (x2);
\draw[-, thick, blue, bend right=30] (z2) to (x2);
\node (D2) at (3.5,2) [circle, draw, inner sep=0, minimum size=10pt] {};
\draw[-] (x2) to (D2);
\node (E) at (4,0.5) [circle, draw, inner sep=0, minimum size=10pt] {};
\draw[-, thick, green!60!black] (D) to (E);
\node (A3) at (4.5,0.5) [circle, draw, inner sep=0, minimum size=10pt] {};
\draw[-] (E) to (A3);
\node (u2) at (4.5,1) [circle, draw, inner sep=0, minimum size=10pt] {\small{$u$}};
\draw[-] (A3) to (u2);
\node (B2) at (4.5,1.5) [circle, draw, inner sep=0, minimum size=10pt] {};
\draw[-] (u2) to (B2);
\node (A4) at (4.5,2) [circle, draw, inner sep=0, minimum size=10pt] {};
\draw[-, thick, green!60!black] (B2) to (A4);
\node (E2) at (4,2) [circle, draw, inner sep=0, minimum size=10pt] {};
\draw[-] (E2) to (A4);
\draw[-, thick, green!60!black] (E2) to (D2);
\draw[dashed] (0.5, 1.7) to (4.5, 1.7);
\draw[dashed] (3.75, 0.5) to (3.75, 2);
\node at (0.75,1.85) {\small{$C$}};
\node at (1,1) {\small{$D$}};
\node at (4,1) {\small{$A$}};
\node at (4.25,1.85) {\small{$B$}};
\end{tikzpicture}
\end{center}
\end{minipage}
\begin{minipage}{0.5\textwidth}
\begin{center}
\begin{tikzpicture}[scale=1.4]
\node [ellipse, minimum width=175pt, minimum height=107pt, draw, thick] at (0,0) {}; 
\draw[-, bend left, thick] (-1,0.05) to (1,0.05);
\draw[-, bend right, thick] (-1.1,0.1) to (1.1,0.1);
\node (A) at (0.8,-0.6) [circle, draw, inner sep=0, minimum size=10pt] {\small{$y$}};
\node (B) at (1.5,0) [circle, draw, inner sep=0, minimum size=10pt] {\small{$z$}};
\draw[-, bend right=10] (A) to (B);
\node (C) at (0.8, 0.5) [circle, draw, inner sep=0, minimum size=10pt] {\small{$v$}};
\draw[-, bend right=10] (B) to (C);
\node (D) at (0, 0.65) [circle, draw, inner sep=0, minimum size=10pt] {};
\node (E) at (0, 1.05) [circle, draw, inner sep=0, minimum size=10pt] {};
\node (Eh) at (0.45, 0.9) [circle, draw, inner sep=0, minimum size=10pt] {\small{$u$}};
\draw[-, bend right=10] (C) to (D);
\draw[-] (D) to (Eh);
\draw[-] (E) to (Eh);
\draw[thick, green!60!black] (0,1.16) arc(39:90:0.16cm and 0.49cm);
\draw[green!60!black, dashed, thick] (-0.13, 1.33) arc(90:270:0.16cm and 0.49cm);
\draw[thick, green!60!black] (-0.13, 0.34) arc(270:323:0.16cm and 0.49cm);
\node (F) at (-0.8, 0.5) [circle, draw] {};
\draw[-, bend right=10]   (D) to (F);
\node (G) at (-1.5, 0) [circle, draw, inner sep=0, minimum size=10pt] {};
\draw[-, bend right=10, thick, green!60!black]   (F) to (G);
\node (H) at (-0.8, -0.6) [circle, draw, inner sep=0, minimum size=10pt] {\small{$x$}};
\node (I) at (0, -0.8) [circle, draw, inner sep=0, minimum size=10pt] {};
\draw[-, bend right=10] (G) to (H);
\draw[-, bend right=10] (H) to (I);
\draw[-, bend right=10] (I) to (A);
\draw[-, bend right=0, blue, thick] (B) to (H);
\draw[-, bend left=60, blue, thick] (A) to (H);
\draw[-, thick, blue] (0.45, 1) arc(153:90:0.095cm and 0.56cm);
\draw[-, thick, blue] (0.8, 0.39) arc(-39:-90:0.095cm and 0.56cm);
\draw[-, thick, blue, dashed] (0.54, 1.3) arc(90:-4:0.095cm and 0.56cm);
\draw[-, thick, blue, dashed] (0.72, 0.2) arc(-90:-174:0.095cm and 0.56cm);
\node  at (Eh) [circle, draw, inner sep=0, minimum size=10pt, fill=white] {\small{$u$}};
\node at(C)  [circle, draw, inner sep=0, minimum size=10pt, fill=white] {\small{$v$}};
\end{tikzpicture}
\end{center}
\begin{center}
\begin{tikzpicture}[scale=1.3]
\draw[-, thick] (-0.5,0) to (4,0) [postaction={decorate,decoration={markings,
        mark=between positions 0.49 and 0.52 step 0.02 with {\arrow[thick]{>};}
        }}];
\draw[-, thick] (-0.5,0) to (-0.5,2) [postaction={decorate,decoration={markings,
        mark=at position 0.5 with {\arrow[thick]{>};}
        }}];
\draw[-, thick] (4,0) to (4,2) [postaction={decorate,decoration={markings,
        mark=at position 0.5 with {\arrow[thick]{>};}
        }}];
\draw[-, thick] (-0.5,2) to (4,2) [postaction={decorate,decoration={markings,
        mark=between positions 0.49 and 0.52 step 0.02 with {\arrow[thick]{>};}
        }}];
\node (A) at (0.5,0.5) [circle, draw, inner sep=0, minimum size=10pt] {};
\node (u) at (0.5,1)[circle, draw, inner sep=0, minimum size=10pt] {\small{$u$}};
\node (B) at (0.5,1.5)[circle, draw, inner sep=0, minimum size=10pt] {};
\draw[-] (A) to (u);
\draw[-] (B) to (u);
\node (v) at (1, 0.5)[circle, draw, inner sep=0, minimum size=10pt] {\small{$v$}};
\draw[-, thick, blue, bend right=5] (u) to (0.9,2);
\draw[- , thick, blue] (0.9,0) to (v);
\draw[-] (A) to (v);
\node (z) at (1.5, 0.5)[circle, draw, inner sep=0, minimum size=10pt] {\small{$z$}};
\draw[-] (v) to (z);
\node (y) at (2, 0.5)[circle, draw, inner sep=0, minimum size=10pt] {\small{$y$}};
\draw[-] (z) to (y);
\node (C) at (2.5, 0.5)[circle, draw, inner sep=0, minimum size=10pt] {};
\draw[-] (y) to (C);
\node (x) at (3, 0.5)[circle, draw, inner sep=0, minimum size=10pt] {\small{$x$}};
\draw[-] (C) to (x);
\draw[-, thick, blue, bend left=35] (y) to (x);
\draw[-, thick, blue, bend right=45] (z) to (x);
\node (D) at (3.5,0.5) [circle, draw, inner sep=0, minimum size=10pt] {};
\draw[-] (x) to (D);
\draw[-, thick, green!60!black] (D) to (4,0.5);
\draw[-, thick, green!60!black] (B) to (0.5,2);
\draw[-, thick, green!60!black] (A) to (0.5,0);
\node (E) at (0,0.5) [circle, draw, inner sep=0, minimum size=10pt] {};
\draw[-] (A) to (E);
\draw[-, thick, green!60!black] (-0.5,0.5) to (E);
\node at (-0.25, 0.25) {\small{$B$}};
\node at (-0.25, 1.5) {\small{$A$}};
\node at (3.25, 0.25) {\small{$C$}};
\node at (3.25, 1.5) {\small{$D$}};
\end{tikzpicture}
\end{center}
\begin{center}
\begin{tikzpicture}[scale=1.3]
\node (A) at (0.5,0.5) [circle, draw, inner sep=0, minimum size=10pt] {};
\node (u) at (0.5,1)[circle, draw, inner sep=0, minimum size=10pt] {\small{$u$}};
\node (B) at (0.5,1.5)[circle, draw, inner sep=0, minimum size=10pt] {};
\draw[-] (A) to (u);
\draw[-] (B) to (u);
\node (v) at (1, 0.5)[circle, draw, inner sep=0, minimum size=10pt] {\small{$v$}};
\draw[-] (A) to (v);
\node (z) at (1.5, 0.5)[circle, draw, inner sep=0, minimum size=10pt] {\small{$z$}};
\draw[-] (v) to (z);
\node (y) at (2, 0.5)[circle, draw, inner sep=0, minimum size=10pt] {\small{$y$}};
\draw[-] (z) to (y);
\node (C) at (2.5, 0.5)[circle, draw, inner sep=0, minimum size=10pt] {};
\draw[-] (y) to (C);
\node (x) at (3, 0.5)[circle, draw, inner sep=0, minimum size=10pt] {\small{$x$}};
\draw[-] (C) to (x);
\node (D) at (3.5,0.5) [circle, draw, inner sep=0, minimum size=10pt] {};
\draw[-] (x) to (D);
\node (A2) at (0.5,2) [circle, draw, inner sep=0, minimum size=10pt] {};
\draw[-, thick, green!60!black] (B) to (A2);
\node (v2) at (1, 2)[circle, draw, inner sep=0, minimum size=10pt] {\small{$v$}};
\draw[-] (A2) to (v2);
\node (z2) at (1.5, 2)[circle, draw, inner sep=0, minimum size=10pt] {\small{$z$}};
\draw[-] (v2) to (z2);
\node (y2) at (2, 2)[circle, draw, inner sep=0, minimum size=10pt] {\small{$y$}};
\draw[-] (z2) to (y2);
\node (C2) at (2.5, 2)[circle, draw, inner sep=0, minimum size=10pt] {};
\draw[-] (y2) to (C2);
\node (x2) at (3, 2)[circle, draw, inner sep=0, minimum size=10pt] {\small{$x$}};
\draw[-] (C2) to (x2);
\draw[-, thick, blue, bend left=30] (y) to (x);
\draw[-, thick, blue, bend right=30] (z2) to (x2);
\node (D2) at (3.5,2) [circle, draw, inner sep=0, minimum size=10pt] {};
\draw[-] (x2) to (D2);
\node (E) at (4,0.5) [circle, draw, inner sep=0, minimum size=10pt] {};
\draw[-, thick, green!60!black] (D) to (E);
\node (A3) at (4.5,0.5) [circle, draw, inner sep=0, minimum size=10pt] {};
\draw[-] (E) to (A3);
\node (u2) at (4.5,1) [circle, draw, inner sep=0, minimum size=10pt] {\small{$u$}};
\draw[-] (A3) to (u2);
\node (B2) at (4.5,1.5) [circle, draw, inner sep=0, minimum size=10pt] {};
\draw[-] (u2) to (B2);
\node (A4) at (4.5,2) [circle, draw, inner sep=0, minimum size=10pt] {};
\draw[-, thick, green!60!black] (B2) to (A4);
\node (E2) at (4,2) [circle, draw, inner sep=0, minimum size=10pt] {};
\draw[-] (E2) to (A4);
\draw[-, thick, green!60!black] (E2) to (D2);
\draw[dashed] (0.5, 1.7) to (4.5, 1.7);
\draw[dashed] (3.75, 0.5) to (3.75, 2);
\node at (0.75,1.85) {\small{$C$}};
\node at (1,1) {\small{$D$}};
\node at (4,1) {\small{$A$}};
\node at (4.25,1.85) {\small{$B$}};
\draw[-, thick, blue] (u) to (v2);
\end{tikzpicture}
\end{center}
\end{minipage}
\caption{The two columns show the construction of $D$ for the same graph with two different embeddings. The black edges belong to $T$, the green edges to $D \setminus T$, and the blue edges to $G \setminus D$. In each column, the first subfigure shows the embedding on the torus. The second subfigure shows a different representation: The torus is obtained by gluing together the opposite sides of the rectangle.
The last subfigure shows the disk obtained by cutting the surface along $D$. Note that it contains every edge of $D$ twice and therefore, every vertex up to 4 times. However, note that the embedding specifies between which copies of the vertices the blue edges have to be drawn.  The capital letters $A,B,C,D$ denote areas of the torus and are included for better orientation: Leaving area $A$ to the left leads to area $D$, leaving $A$ to the top leads to $B$ and so on.  } \label{fig:construction}
\end{figure}
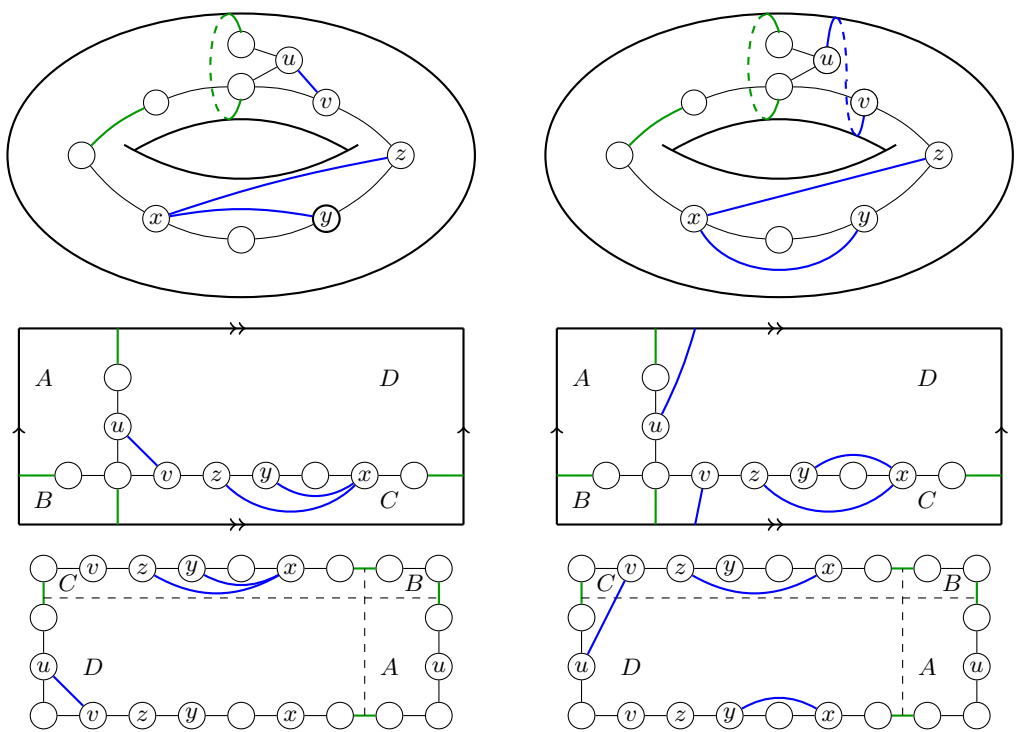

Now, we have all the prerequisites in place to prove \cref{thm:genusspanner}. 
The main idea is to give a similar construction as in \cref{lem:facialcycles} to partition the greedy spanner into facial cycles. Before delving into the proof, let us briefly comment on why \cref{lem:constructionD} is not a reduction to the planar case, i.e., we cannot use the same construction as in \cref{lem:facialcycles}.

Recall that the key ingredient of \cref{lem:facialcycles} was to define a partial order in which an edge~$e'$ precedes another edge $e$ if $e'$ is embedded on the inside of the cycle that $e$ closes with $\mst$. In the bounded genus case, if the cycle closed by $e$ is non-separating, there is no such thing as ``the inside'' of the cycle. For example, consider the edge $(u,v)$ in the right column of \cref{fig:construction} and the cycle it closes with $\mst$. This cycle does not have an ``inside'' and cannot be decomposed into multiple faces. In particular, the cycle disappears after cutting the surface along $D$. However, it separates the disk bounded by $D$ into two parts. Therefore, we  have to consider cycles that include edges of $D \setminus \mst$.

\begin{proof}[Proof of \cref{thm:genusspanner}]

Let $G$ be some graph of genus $g$. Let $H$ be the greedy $(1+\epsilon)$-spanner of $G$ and let MST denote a minimum spanning tree of $H$. By \cref{greedyMST}, we know that $\mst$ is also a minimum spanning tree of $G$, so that it suffices to show ${w(H)\leq \bigl( 1+\frac{2}{\epsilon}\bigl) \bigl(1+\frac{2g}{1+\epsilon}\bigl)w(\mst)}$.

Let $g'$ be the genus of $H$. If $g'=0$, the assertion follows directly by \cref{thm:planarspanner}. Therefore, we assume from now on $g'\geq 1$. Note that $g' \leq g$ because $H$ is a subgraph of $G$. Fix an embedding of $H$ on an orientable closed surface of genus $g'$ such that no two edges cross. 
By \cref{lem:constructionD}, there is a subgraph $D$ of $H$ with $\mst \subseteq D$ such that
 \begin{equation}\label{eq:2g}
 |E(D)\setminus E(\mst)| \leq 2g' \leq 2g
 \end{equation}  and such that the edges of $D$ induce only one face and bound a topological disk. 
Next, observe that, for every edge $e$ in $H\setminus \mst$, we have $w(e)\leq w(\mst)/(1+\epsilon)$: Every edge $e$ in $H \setminus \mst$ closes a cycle $C$ together with the edges of $\mst$. Using \cref{obs:longcycle}, we obtain
\begin{equation*}
w(e) < \frac{w(C\setminus \{e\})}{1+\epsilon} \leq \frac{w(\mst)}{1+\epsilon}.
\end{equation*} 
In particular, this is fulfilled for edges in $D \setminus \mst$. Combining this with \eqref{eq:2g}, we obtain
\begin{equation}\label{eq:weightofD}
w(D) \leq \left( 1+ \frac{2g}{1+\epsilon} \right) w(\mst).
\end{equation}
The next step is to bound the weight of $H$ by $(1+2/\epsilon)w(D)$. For this, we use a similar construction as in \cref{lem:facialcycles} and show that it is possible to iteratively choose an edge $e$ in $H \setminus D$ which, together with the edges of $D$ and the edges chosen in previous iterations, closes a facial cycle $C_e$ in the embedding of $H$.

In each iteration, we find a suitable edge as follows: Pick an arbitrary \hbox{edge $e$ of $H\setminus D$}. 
If it defines a facial cycle together with $D$ and edges chosen in previous iterations, we can simply choose $e$. 
Assume this is not the case. Note that $e$ cuts the disk bounded by $D$ in two parts and both contain edges in $H \setminus D$ to which no cycles have been assigned yet (otherwise~$e$ would close a suitable facial cycle). 
Pick the part whose boundary with $D$ contains fewer edges (breaking ties arbitrarily) and pick a new edge $e'$ in $H \setminus D$ which lies inside this half and has not yet been chosen in a previous iteration. Note that $e'$ again cuts the disk in two parts and the boundary of the smaller part contains fewer edges of $D$ than in the step before.  Therefore, by repeating the steps above, we will end up with a suitable edge after finitely many steps.  For example, on the left side of \cref{fig:construction}, if we pick $e=(x,z)$, we will set $e'=(x,y)$ and this edge is suitable. After this, we can assign a facial cycle to $(x,z)$ and then to $(u,v)$. In the instance on the right, we can assign the cycles to the blue edges in any order.

Note that, in this construction, no two edges are assigned the same facial cycle.
 As every edge is contained in at most two facial cycles, we have
\begin{equation}\label{eq:facecycles}
\sum\limits_{e \in H \setminus D} w(C_e) \leq 2 w(H).
\end{equation}
Therefore, we can now apply \cref{lem:spannercomb} and obtain
\begin{equation*}
w(H) \overset{\text{Lem \ref{lem:spannercomb}}}{\leq} \left(1+\frac{2}{\epsilon} \right) w(D)
\overset{\eqref{eq:weightofD}}{\leq} \left(1+\frac{2}{\epsilon} \right) \left(1+\frac{2g}{1+\epsilon} \right) w(\mst). \qedhere
\end{equation*}

\end{proof}

Recall that Grigni showed that every graph of genus $g\geq 1$ contains a $(1+\epsilon)$-spanner of lightness at most  $1+(12g-4)/\epsilon$ \cite{grigni}.
Let us briefly comment on how our bound compares to Grigni's bound. For this, note that, for $g\geq 1$,
\begin{equation*}
 \left( 1+\frac{2}{\epsilon} \right) \left( 1+ \frac{2g}{1+\epsilon} \right)
 = 1 + \frac{2}{\epsilon} + \frac{2g}{1+\epsilon} + \frac{4g}{\epsilon (1+\epsilon)} 
 < 1 + \frac{2g}{\epsilon} + \frac{2g}{\epsilon} + \frac{4g}{\epsilon} = 1+\frac{8g}{\epsilon}.
\end{equation*}
Therefore, our bound is stronger than Grigni's bound for every $g\geq 1$. Moreover, in \linebreak the  planar case (i.e., $g=0$), we obtain a lightness of $1+\frac{2}{\epsilon}$. It was shown by \linebreak Althöfer et al.~\cite[Theorem 5]{althoefer} that this is best possible, i.e., our bound is tight for planar graphs.
Note that the worst-case lightness for spanners of graphs of genus $g$ has to increase in $g$, since not every graph admits a light spanner.
For example, for every $k\geq 3$ and almost all $n$, there is a graph on $n$ vertices with girth at least~$k$ and at least $\frac{1}{4}n^{1+\frac{1}{k}}$ edges~\hbox{\cite[Theorem 6.6]{mitzenmacher}}.

\section{Open problems}
The key question in online graph exploration is whether the problem admits a constant-competitive algorithm~\cite{pruhs94}. 
While this problem remains open, our results suggest steps that might be needed towards a resolution of this question.
Firstly, we have shown that the online graph exploration problem allows for a constant-competitive algorithm on graphs admitting a light spanner, in particular, minor-free graphs. 
This suggests that, for proving a non-constant general lower bound on the competitive ratio, one might require dense high-girth graphs or expanders \cite{krivelevich}.
Not even a competitive ratio of $o(\log (n))$ has yet been attained, and our results eliminate $\blocking_\delta$, for most values of~$\delta$, as a candidate for achieving this.
It remains to close the gap between $\delta \in o(\log (n)/ \log \log (n))$ and $\delta \in \Omega(\log (n))$.

Regarding spanners, we gave an improved upper bound on the lightness of spanners in bounded genus graphs.
It is a natural question whether our bound is already tight for $g\geq 1$ or can further be improved.
In particular, it is unclear whether the worst-case lightness for a fixed stretch must depend linearly on $g$.

\newpage
\bibliography{ExplorationOfGraphsWithExcludedMinors}

\begin{thebibliography}{10}

\bibitem{albers}
Susanne Albers and Monika~R. Henzinger.
\newblock Exploring unknown environments.
\newblock {\em SIAM Journal on Computing}, 29(4):1164--1188, 2000.

\bibitem{althoefer}
Ingo Alth{\"{o}}fer, Gautam Das, David~P. Dobkin, Deborah Joseph, and
  Jos{\'{e}} Soares.
\newblock On sparse spanners of weighted graphs.
\newblock {\em Discrete \& Computational Geometry}, 9:81--100, 1993.

\bibitem{arora}
Sanjeev Arora, Michelangelo Grigni, David~R. Karger, Philip~N. Klein, and
  Andrzej Woloszyn.
\newblock A polynomial-time approximation scheme for weighted planar graph
  {TSP}.
\newblock In {\em Proceedings of the 9th Annual ACM-SIAM Symposium on Discrete
  Algorithms (SODA)}, pages 33--41, 1998.

\bibitem{distcomp}
Baruch Awerbuch, Alan Baratz, and David Peleg.
\newblock Cost-sensitive analysis of communication protocols.
\newblock In {\em Proceedings of the 9th Annual ACM Symposium on Principles of
  Distributed Computing (PODC)}, page 177–187, 1990.

\bibitem{birx}
Alexander Birx, Yann Disser, Alexander~V. Hopp, and Christina Karousatou.
\newblock An improved lower bound for competitive graph exploration.
\newblock {\em Theoretical Computer Science}, 868:65--86, 2021.

\bibitem{bjelde}
Antje Bjelde, Jan Hackfeld, Yann Disser, Christoph Hansknecht, Maarten Lipmann,
  Julie Mei{\ss}ner, Miriam Schl{\"{o}}ter, Kevin Schewior, and Leen Stougie.
\newblock Tight bounds for online {TSP} on the line.
\newblock {\em {ACM} Transactions on Algorithms}, 17(1):3:1--3:58, 2021.

\bibitem{bonifaci}
Vincenzo Bonifaci and Leen Stougie.
\newblock Online k-server routing problems.
\newblock {\em Theory of Computing Systems}, 45(3):470--485, 2009.

\bibitem{borradaile}
Glencora Borradaile, Erik~D. Demaine, and Siamak Tazari.
\newblock Polynomial-time approximation schemes for subset-connectivity
  problems in bounded-genus graphs.
\newblock {\em Algorithmica}, 68(2):287--311, 2014.

\bibitem{minorspanner}
Glencora Borradaile, Hung Le, and Christian Wulff-Nilsen.
\newblock Minor-free graphs have light spanners.
\newblock In {\em Proceedings of the 58th Annual Symposium on Foundations of
  Computer Science (FOCS)}, pages 767--778, 2017.

\bibitem{tadpole}
Sebastian Brandt, Klaus-Tycho Foerster, Jonathan Maurer, and Roger Wattenhofer.
\newblock Online graph exploration on a restricted graph class: Optimal
  solutions for tadpole graphs.
\newblock {\em Theoretical Computer Science}, 839:176--185, 2020.

\bibitem{chechik}
Shiri Chechik and Christian Wulff-Nilsen.
\newblock Near-optimal light spanners.
\newblock {\em ACM Transactions on Algorithms}, 14(3):1--15, 2018.

\bibitem{demainetreewidth}
Erik~D. Demaine, MohammadTaghi HajiAghayi, and Bojan Mohar.
\newblock Approximation algorithms via contraction decomposition.
\newblock {\em Combinatorica}, 30(5):533--552, 2010.

\bibitem{deng}
Xiaotie Deng and Christos~H. Papadimitriou.
\newblock Exploring an unknown graph.
\newblock {\em Journal of Graph Theory}, 32(3):265--297, 1999.

\bibitem{dereniowski}
Dariusz Dereniowski, Yann Disser, Adrian Kosowski, Dominik Paj, and
  Przemys{\l}aw Uzna{\'n}ski.
\newblock Fast collaborative graph exploration.
\newblock {\em Information and Computation}, 243:37--49, 2015.

\bibitem{DisserHackfeldKlimm/19}
Yann Disser, Jan Hackfeld, and Max Klimm.
\newblock Tight bounds for undirected graph exploration with pebbles and
  multiple agents.
\newblock {\em Journal of the ACM}, 66(6):40(41), 2019.

\bibitem{disserlower}
Yann Disser, Frank Mousset, Andreas Noever, Nemanja Skoric, and Angelika
  Steger.
\newblock A general lower bound for collaborative tree exploration.
\newblock {\em Theoretical Computer Science}, 811:70--78, 2020.

\bibitem{dobrev}
Stefan Dobrev, Rastislav Kr{\'a}lovi{\v{c}}, and Euripides Markou.
\newblock Online graph exploration with advice.
\newblock In {\em Proceedings of the 19th International Colloquium on
  Structural Information and Communication Complexity (SIROCCO)}, pages
  267--278, 2012.

\bibitem{eberle}
Franziska Eberle, Alexander Lindermayr, Nicole Megow, Lukas N{\"o}lke, and Jens
  Schl{\"o}ter.
\newblock Robustification of online graph exploration methods.
\newblock In {\em Proceedings of the 36th Conference on Artificial Intelligence
  (AAAI)}, pages 9732--9740, 2022.

\bibitem{erdos}
Paul Erdős.
\newblock Extremal problems in graph theory.
\newblock In {\em Proceedings of the Symposium on Theory of Graphs and its
  Applications}, pages 29--36, 1963.

\bibitem{filtser}
Arnold Filtser and Shay Solomon.
\newblock The greedy spanner is existentially optimal.
\newblock {\em {SIAM} Journal on Computing}, 49(2):429--447, 2020.

\bibitem{fleischer}
Rudolf Fleischer and Gerhard Trippen.
\newblock Exploring an unknown graph efficiently.
\newblock In {\em Proceedings of the 13th Annual European Symposium on
  Algorithms (ESA)}, pages 11--22, 2005.

\bibitem{foerster}
Klaus-Tycho Foerster and Roger Wattenhofer.
\newblock Lower and upper competitive bounds for online directed graph
  exploration.
\newblock {\em Theoretical Computer Science}, 655:15--29, 2016.

\bibitem{fritsch}
Robin Fritsch.
\newblock Online graph exploration on trees, unicyclic graphs and cactus
  graphs.
\newblock {\em Information Processing Letters}, 168:106096, 2021.

\bibitem{grigni}
Michelangelo Grigni.
\newblock Approximate {TSP} in graphs with forbidden minors.
\newblock In {\em Proceedings of the 27th International Colloquium on Automata,
  Languages, and Programming (ICALP)}, volume 1853, pages 869--877, 2000.

\bibitem{grignipathwidth}
Michelangelo Grigni and Hao{-}Hsiang Hung.
\newblock Light spanners in bounded pathwidth graphs.
\newblock In {\em Proceedings of the 37th International Symposium on
  Mathematical Foundations of Computer Science (MFCS)}, pages 467--477, 2012.

\bibitem{grigniapex}
Michelangelo Grigni and Papa Sissokho.
\newblock Light spanners and approximate {TSP} in weighted graphs with
  forbidden minors.
\newblock In {\em Proceedings of the 13th ACM-SIAM Symposium on Discrete
  Algorithms (SODA)}, pages 852--857, 2002.

\bibitem{hatcher}
Allen Hatcher.
\newblock {\em {Algebraic topology}}.
\newblock Cambridge University Press, 2000.
\newblock URL: \url{https://cds.cern.ch/record/478079}.

\bibitem{hurkens}
Cor~A.J. Hurkens and Gerhard~J. Woeginger.
\newblock On the nearest neighbor rule for the traveling salesman problem.
\newblock {\em Operations Research Letters}, 32(1):1--4, 2004.

\bibitem{pruhs94}
Bala Kalyanasundaram and Kirk~R. Pruhs.
\newblock Constructing competitive tours from local information.
\newblock {\em Theoretical Computer Science}, 130(1):125--138, 1994.

\bibitem{krivelevich}
Michael Krivelevich and Benjamin Sudakov.
\newblock Minors in expanding graphs.
\newblock {\em Geometric and Functional Analysis}, 19(1):294--331, 2009.

\bibitem{megow}
Nicole Megow, Kurt Mehlhorn, and Pascal Schweitzer.
\newblock Online graph exploration: New results on old and new algorithms.
\newblock {\em Theoretical Computer Science}, 463:62--72, 2012.

\bibitem{mitzenmacher}
Michael Mitzenmacher and Eli Upfal.
\newblock {\em Probability and computing: Randomization and probabilistic
  techniques in algorithms and data analysis}.
\newblock Cambridge University Press, 2017.

\bibitem{miyazaki}
Shuichi Miyazaki, Naoyuki Morimoto, and Yasuo Okabe.
\newblock The online graph exploration problem on restricted graphs.
\newblock {\em IEICE transactions on information and systems},
  92(9):1620--1627, 2009.

\bibitem{mohar}
Bojan Mohar and Carsten Thomassen.
\newblock {\em Graphs on Surfaces}.
\newblock Johns Hopkins University Press, 2001.

\bibitem{peleg}
David Peleg and Alejandro~A. Sch{\"a}ffer.
\newblock Graph spanners.
\newblock {\em Journal of graph theory}, 13(1):99--116, 1989.

\bibitem{NN}
Daniel~J. Rosenkrantz, Richard~E. Stearns, and Philip~M. Lewis, II.
\newblock An analysis of several heuristics for the traveling salesman problem.
\newblock {\em SIAM journal on computing}, 6(3):563--581, 1977.

\bibitem{biology}
Daniel Russel and Leonidas~J. Guibas.
\newblock Exploring protein folding trajectories using geometric spanners.
\newblock In {\em Pacific Symposium on Biocomputing (PSB)}, pages 42--53. World
  Scientific, 2005.

\bibitem{routing}
Bang~Ye Wu, Kun{-}Mao Chao, and Chuan~Yi Tang.
\newblock Light graphs with small routing cost.
\newblock {\em Networks}, 39(3):130--138, 2002.

\bibitem{youngs63}
John William~Theodore Youngs.
\newblock Minimal imbeddings and the genus of a graph.
\newblock {\em Journal of Mathematics and Mechanics}, pages 303--315, 1963.

\end{thebibliography}
\newpage

\end{document}